\documentclass[11pt]{article}
\usepackage{amsmath}
\usepackage{amsfonts}
\usepackage{amsthm}
\usepackage{amsmath,amssymb}
\usepackage{graphicx}
\usepackage{color}
\usepackage{url}
\usepackage{graphicx}
\usepackage{appendix}
\usepackage{booktabs}

\usepackage{makecell}
\usepackage{multirow}

\renewcommand{\arraystretch}{1.2}
\allowdisplaybreaks[4]

\newtheorem{theorem}{Theorem}
\newtheorem{lemma}{Lemma}

\newtheorem{definition}{Definition}

\newtheorem{remark}{Remark}

\newcommand{\RNum}[1]{\uppercase\expandafter{\romannumeral #1\relax}}

\newcommand{\ftwon}{{\mathbb F}_{2^n}}

\newcommand{\ftwo}{{\mathbb F}_{2}}
\newcommand{\ls}[1]
    {\dimen0=\fontdimen6\the\font\lineskip=#1\dimen0
     \advance\lineskip.5\fontdimen5\the\font
     \advance\lineskip-\dimen0
     \lineskiplimit=0.9\lineskip
     \baselineskip=\lineskip
     \advance\baselineskip\dimen0
     \normallineskip\lineskip\normallineskiplimit\lineskiplimit
     \normalbaselineskip\baselineskip
     \ignorespaces}

\begin{document}

\title{A new class of S-boxes with optimal Feistel boomerang uniformity}
\author{Yuxuan Lu\footnotemark[1]\thanks{Y. Lu, N. Li,  and L. Wang are with the Hubei Provincial Engineering Research Center of Intelligent Connected Vehicle Network Security, School of Cyber Science and Technology, Hubei University, Wuhan 430062, China. Email: yuxuan.lu@aliyun.com, nian.li@hubu.edu.cn, wangtaolisha@163.com},
 Sihem Mesnager\footnotemark[2]\thanks{S. Mesnager is with the Department of Mathematics, University of Paris VIII, F-93526 Saint-Denis, Paris, France,
the Laboratory of Analysis, Geometry, and Applications (LAGA), University Sorbonne Paris Nord CNRS,
UMR 7539, F-93430, Villetaneuse, France, and also with the T\'el\'ecom Paris, Palaiseau, France. Email:
smesnager@univ-paris8.fr},
Nian Li\footnotemark[1],
Lisha Wang\footnotemark[1],
Xiangyong Zeng \footnotemark[3]\thanks{X. Zeng is with Hubei Key Laboratory of Applied Mathematics, Faculty of Mathematics and Statistics, Hubei University, Wuhan 430062, China. Email: xzeng@hubu.edu.cn}}
\date{\today}%
\maketitle
\begin{abstract}

The Feistel Boomerang Connectivity Table ($\rm{FBCT}$), which is the Feistel version of the Boomerang Connectivity Table ($\rm{BCT}$), plays a vital role in analyzing block ciphers' ability to withstand strong attacks, such as boomerang attacks. However, as of now, only four classes of power functions are known to have explicit values for all entries in their $\rm{FBCT}$. In this paper, we focus on studying the FBCT of the power function $F(x)=x^{2^{n-2}-1}$ over $\mathbb{F}_{2^n}$, where $n$ is a positive integer. Through certain refined manipulations to solve specific equations over $\mathbb{F}_{2^n}$ and employing binary Kloosterman sums, we determine explicit values for all entries in the  $\rm{FBCT}$ of $F(x)$ and further analyze its Feistel boomerang spectrum. Finally,  we demonstrate that this power function exhibits the lowest Feistel boomerang uniformity.

\medskip

 \medskip


\end{abstract}

\section{Introduction}

In symmetric cryptography, block ciphers use S-boxes (substitution boxes) that take $n$ binary inputs and produce an $m$-bit output, where $n$ and $m$ are positive integers. The S-box is the main nonlinear component of the cryptographic algorithm and plays a crucial role in enhancing its security. One of the most potent attacks in symmetric-key cryptography is the differential cryptanalysis attack, introduced by Biham and Shamir in 1991. Differential cryptanalysis is fundamental for evaluating the security of block ciphers. The ability of a cryptographic algorithm to resist differential attacks is closely linked to the resistance of the S-box it uses. In 1993, Nyberg introduced the Differential Distribution Table ($\rm {DDT}$) and differential uniformity to measure an S-box's resistance to differential attacks. The smaller the differential uniformity of an S-box, the stronger its resistance to these attacks. An S-box with a differential uniformity of 2 is considered almost perfect nonlinear ($\rm {APN}$).

The boomerang attack, introduced by Wanger (\cite{WD.1999}) in 1999, is a variant of the differential attack that is an essential cryptographic analysis technique. At Eurocrypt 2018, Cid et al.  (\cite{CC.HT2018}) improved the analysis of boomerang-style attacks by introducing the Boomerang Connectivity Table ($\rm{BCT}$). To measure a function's resistance against boomerang attacks, Boura and Canteaut  (\cite{BC.CA2018}) introduced the concept of boomerang uniformity, similar to the resistance against differential attacks. Most previous research has focused on Substitution-Permutation Network (SPN) structures and has primarily overlooked ciphers following the Feistel Network model. It is important to address Feistel Network ciphers as well due to their significant practical applications, such as 3-DES and CLEFIA. To fill this gap, Boukerrou et al.  (\cite{BH.HP2020}) extended the Boomerang Connectivity Table (BCT) to accommodate S-boxes within Feistel Network ciphers, even when these S-boxes are not permutations and introduced the Feistel Boomerang Connectivity Table ($\rm{FBCT}$).

The concise and efficient representation of power functions, especially in hardware, has attracted much attention. The Feistel boomerang distinguisher ($\rm{FBCT}$) calculation for these functions has been a significant area of recent research. In a paper by Eddahmani and  Mesnager  (\cite{ES.MS2022}), the values of the entries of the FBCT for the inverse, Gold, and Bracken-Leander functions over finite fields with even characteristics were fully specified. The authors also estimated the number of elements $(a,b)\in\ftwon^2$ with potential values in the FBCT. In 2023, Man \cite{MY.MS2023} computed the specific values of all entries in the $\rm{FBCT}$ for a Niho type power function $F(x)=x^{2^{m+1}-1}$ over $\ftwon$ for $n=2m$. Generally, it is challenging to determine the explicit values of all entries in the  $\rm{FBCT}$ and the Feistel boomerang spectrum for a given function. For further insights on the Feistel boomerang uniformity of certain power functions over $\ftwon$, readers are directed to \cite{GK.HS2309}, \cite{GK.HS2310} , and  \cite{MY.LN2023}.  Boukerrou et al. have also demonstrated in \cite{BH.HP2020} that ${\rm FBCT}_F(a,b)\equiv0\ ({\rm mod}\ 4)$ for all $a,b\in\ftwon$, and all the non-trivial entries of the  $\rm{FBCT}$ of a function $F$ over $\ftwon$ are 0 if and only if it is APN. Consequently, the minimum value of the Feistel boomerang uniformity for a non-APN function is $4$. Table \ref{table-1} lists power functions with known Feistel boomerang uniformity, excluding APN power functions.


In this paper, we explore the Feistel boomerang properties of a certain class of power functions represented as $F(x)=x^{2^{n-2}-1}$ over $\mathbb{F}_{2^n}$. By employing specific techniques for solving equations over finite fields and the binary Kloosterman sum, we can calculate the explicit values of all entries of the   $\rm{FBCT}$ of $F$ and determine the Feistel boomerang spectrum of $F$.
As a main result, we demonstrate that the Feistel boomerang uniformity of $F$ is $4$ if $n$ is not a multiple of $3$  and the uniformity is $8$ when $n$ is a multiple of $3$. Notably, the power function $F(x)=x^{2^{n-2}-1}$ over $\mathbb{F}_{2^n}$ investigated in this paper constitutes the fifth class of functions with the known Feistel boomerang spectrum. Furthermore, it possesses the lowest Feistel boomerang uniformity of $4$ when $3\nmid n$  among the non-APN power functions.

The paper is organized as follows. Section \ref{Preliminaries} presents some background information, basic definitions, and notation.  Section \ref{statement} outlines the main results of the paper, along with detailed proofs in subsections \ref{proof1} and \ref{proof2}.  Finally, Section \ref{conc} provides the paper's conclusion.

\begin{table}[h]
\caption{$F(x)=x^d$ over $\ftwon$ with known Feistel boomerang uniformity}\label{table-1}
\renewcommand\arraystretch{1.1}
\setlength\tabcolsep{10pt}
\centering
\begin{tabular}{cccccc}
		\toprule
     No.& $d$ & Condition & $\tilde{\beta} (F)$ & Spectrum &  Ref. \\
       \midrule
      1  & $2^n-2$ & $n$ even & 4 & known &  \cite{ES.MS2022} \\
     2 & $2^k+1$ & ${\rm gcd}(n,k)=d,d\ne1$ & $2^n$& known& \cite{ES.MS2022}\\
     3 &$2^{2k}+2^k+1$ & $n=4k$ & $2^{2k}$ &known& \cite{ES.MS2022}\\
     4  & $2^{m+1}-1$ & $n=2m$ & $2^m$ &known& \cite{MY.MS2023} \\
     5  & $2^m-1$  & $n=2m+1$ or $n=2m$  &  $2^m-4$ & unknown & \cite{GK.HS2309}  \\
     6  & $2^{\frac{n+3}{2}}-1$ & $n$ odd & $4$ & unknown & \cite{GK.HS2309}   \\
     7  & $21$  & $n$ odd or $n$ even &  $4$ or $16$ & unknown & \cite{GK.HS2310}  \\
     8  & $2^n-2^s$  & ${\rm gcd}(n, s+1)=1$, $n-s=3$ &  $4$ & unknown &  \cite{GK.HS2310} \\
     9  & $2^{m+1}+3$  & $n=2m+1$ or $n=2m$  & $4$ or $2^m$  & unknown & \cite{MY.LN2023}  \\
     10 & $7$ & arbitrary  &  $4$ & unknown & \cite{MY.LN2023}\\
     11 &$2^{n-2}-1$ & {\rm 3} $\nmid n$ or {\rm 3} $\mid n$  & $4$ or $8$ &known & This paper\\
\bottomrule
\end{tabular}
\end{table}

\section{Preliminaries}\label{Preliminaries}

Throughout this paper,  we use the notation $\ftwon$ to represent the finite field with $2^n$ elements, $\ftwon^\ast$ to denote the cyclic group $\ftwon\setminus\{0\}$, and ${\rm Tr}_1^n(\cdot)$ to represent the absolute trace function from ${\mathbb F}_{2^n}$ onto its prime field ${\mathbb F}_{2}$, where $n$ is a positive integer.
We note that for $x\in\ftwon$, ${\rm Tr}_1^n(x)=\sum_{i=0}^{n-1}x^{2^i}$.


Vectorial Boolean functions, also known as multi-output Boolean functions, are functions that map from the finite field $\mathbb{F}_{2^n}$ to $\mathbb{F}_{2^m}$, where $m$ and $n$ are positive integers. These functions are commonly used in the design of block ciphers for cryptography. S-boxes, important components of block cyphers, play a crucial role in symmetric key algorithms for substitution and are significant for the security of several block cyphers.   For complete, comprehensive, and deep developments on vectorial functions for symmetric cryptography,  see the book  \cite{Carlet-book-2021}.

One of the most important attacks on a block cipher is the differential attack, which Biham and Shamir introduced in 1991 \cite{BE.SA1991}.  For a vectorial Boolean function, the tools introduced by Nyberg \cite{Nyberg} in 1993, such as the Difference Distribution Table ($\rm{DDT}$) and the differential uniformity $\delta_F$, are used to study them. The differential uniformity $\delta_F$ of a permutation $F$ (used as an S-box inside a cryptosystem) measures the resistance of the block cipher against differential cryptanalysis. The differential uniformity of a vectorial Boolean function $F: \mathbb{F}_{2^n}\rightarrow \mathbb{F}_{2^n}$ is defined as:
$$
\delta_F=\max_{a,b\in \mathbb{F}_{2^n}, a\neq 0}\rm{DDT}_F(a,b),
$$
where $\rm{DDT}_F(a,b)$ is the entry at $(a,b)\in \left(\mathbb{F}_{2^n}\right)^2$ of the difference distribution table given by:
$$
\rm{DDT}_F(a,b)=|\{x\in \mathbb{F}_{2^n},\ F(x+a)+F(x)=b\}|.
$$

When $F$ is used as an S-box inside a cryptosystem, a smaller value of $\delta_F$ indicates better resistance against a differential attack. Typically, the most optimal functions satisfy $\delta_F=2$ and are called almost perfect nonlinear (APN).

The boomerang attack is an important cryptanalytical technique used on block ciphers, which was introduced as a variant of the technique known as differential cryptanalysis. The boomerang attack can be helpful in situations where no significant differential probability is present for the entire cipher. The resistance of an S-box $F$ (a permutation of $\mathbb{F}_{2^n}$) against boomerang attacks can be measured through its Boomerang Connectivity Table ($\rm{BCT}$) introduced by Cid et al.~\cite{CC.HT2018}. Its boomerang uniformity, denoted by $\beta_F$, is defined as
$$
\beta_F=\max_{a,b\in\mathbb{F}_{2^n},ab\neq 0}\rm{BCT}_F(a,b),
$$
where $\rm{BCT}_F(a,b)$ represents the entry at $(a,b)\in\left(\mathbb{F}_{2^n}\right)^2$ of the Boomerang Connectivity Table of $F$, i.e.,
$$
\rm{BCT}_F(a,b)=|\{x\in\mathbb{F}_{2^n},F^{-1}(F(x)+b)+F^{-1}(F(x+a)+b)=a\}|.
$$

 Li {\em et al.} \cite{LQL19} showed that the BCT table could be defined for vectorial Boolean functions which are not necessarily permutations as follows:
\begin{equation*}
\rm{BCT}_F(a,b)=|\{x,y\in\mathbb F_{2^n}: F(y)+F(x)=b \mbox{ and } F(y+a)+F(x+a)=b\}|.
\end{equation*}
Note that this equivalent formulation does not require the compositional inverse of the function $F$ and enables us to compute the BCT for non-permutations.

The BCT of various families of S-boxes has been studied, and further results on the BCT have been presented for multiple permutations of $\mathbb{F}_{2^n}$ in~\cite{BH.HP2020}. Additionally, the BCT, as presented in~\cite{CC.HT2018}, is valid for a block cipher with a Substitution-Permutation Network (SPN) structure and has been extended to handle S-boxes for block ciphers with a Feistel construction. Boukerrou et al. (\cite{BH.HP2020}) introduced a new tool called the Feistel Boomerang Connectivity Table ($\rm{FBCT}$). The  $\rm{FBCT}$ of a vectorial Boolean function $F : \mathbb{F}_{2^n}\rightarrow\mathbb{F}_{2^n}$ is a $2^n\times 2^n$ table, and the entry at $(a,b)$ is defined as follows.

\begin{definition}\label{definition-FBCT}\rm (\cite{BH.HP2020})
Let $F(x)$ be a mapping from ${\mathbb F}_{2^n}$ to itself. The Feistel Boomerang Connectivity Table ($\rm{FBCT}$) is a $2^n\times2^n$ table defined for $(a, b)\in{\mathbb F}_{2^n}^2$ by
$${\rm FBCT}_F(a, b)=|\{x\in\ftwon: F(x)+F(x+a)+F(x+b)+F(x+a+b)=0\}|.$$
It is clear that the  $\rm{FBCT}$ satisfies ${\rm FBCT}_F(a,b)=2^n$ if $ab(a+b)=0$. Hence, the Feistel boomerang uniformity of $F(x)$ is defined by
$$\tilde{\beta} (F)=\max_{a, b\in \ftwon, ab(a+b)\ne 0} {\rm FBCT}_F(a, b).$$

The Feistel boomerang spectrum is given by the multiset $\{{\rm FBCT}_F(a,b):a,b\in\ftwon\}$.
\end{definition}


\section{Statement of the main result}\label{statement}

\begin{theorem}\label{main-thm}
Let $F$ be the  power function  defined over $\ftwon$  by $F(x)=x^{2^{n-2}-1}$ ($n>6$).
  The two following results hold.
  \begin{itemize}
    \item [\rm (A)] {\bf (The $\rm{FBCT}$ values of $F$)} ${\rm FBCT}_F(a,b)\in\{2^n,0,4\}$ if $3\nmid n$ and ${\rm FBCT}_F(a,b)\in\{2^n,0,4,8\}$ if $3\mid n$ for arbitrary $a,b\in\ftwon$.
    \item [\rm (B)] {\bf (The  Feistel boomerang spectrum)}: For  any $(a,b)$ ranges in $\ftwon^2$, the Feistel boomerang spectrum of $F$ satisfies
\[\begin{tabular}{|c|c|}
 \hline
  ${\rm FBCT}(a,b)$ &{\rm Frequency} \\\hline
    $2^n$  &  $3\cdot2^n-2$ \\\hline
    \multirow{2}{*}{$0$}  &  $\frac{(2^n-1)(3\cdot2^n+3K_n(1)-12)}{4}$, $n$ {\rm odd}\\\cline{2-2}&$\frac{(2^n-1)(3\cdot2^n-3K_n(1)+8)}{4}$, $n$ {\rm even}\\\hline
    \multirow{2}{*}{$4$}  &  $\frac{(2^n-1)(2^n-3K_n(1)+4)}{4}$, $n$ {\rm odd}\\\cline{2-2}&$\frac{(2^n-1)(2^n+3K_n(1)-16)}{4}$, $n$ {\rm even}\\\hline
\end{tabular}\]
if $3\nmid n$, and for $3\mid n$, we have
\[\begin{tabular}{|c|c|}
 \hline
  ${\rm FBCT}(a,b)$ & {\rm Frequency} \\\hline
    $2^n$  &  $3\cdot2^n-2$ \\\hline
    \multirow{2}{*}{$0$}  &  $\frac{(2^n-1)(3\cdot2^n+3K_n(1)-12)}{4}$, $n$ {\rm odd}\\\cline{2-2}&$\frac{(2^n-1)(3\cdot2^n-3K_n(1)+8)}{4}$, $n$ {\rm even}\\\hline
    \multirow{2}{*}{$4$}  &  $\frac{(2^n-1)(2^n-3K_n(1)-20)}{4}$, $n$ {\rm odd}\\\cline{2-2}&$\frac{(2^n-1)(2^n+3K_n(1)-40)}{4}$, $n$ {\rm even}\\\hline
   $8$  &  $6(2^n-1)$ \\\hline
\end{tabular}\]
   \end{itemize}
where $K_n(1)$ is the value of the Kloosterman sum at point 1 that is determined in \cite{Carlitz.1969} as follows (on the assumption that  $\frac{1}{0}:=0$):
  \[K_n(1)=\sum\limits_{x\in\ftwon}(-1)^{{\rm Tr}_1^n(x+x^{-1})}=1+\frac{(-1)^{n-1}}{2^{n-1}}\sum_{i=0}^{\lfloor\frac{n}{2}\rfloor}(-1)^i\binom{n}{2i}7^i.\]

\end{theorem}

\begin{remark}
According to the results obtained by Lachaud and Wolfmann in their 1990 paper \cite{LG.WJ1990}, the Kloosterman sum values in the range $[-2^{\frac{n}{2}+1}+1,2^{\frac{n}{2}+1}+1]$ consist of all multiples of 4. From this, it can be deduced that $2^n-3K_n(1)+4>0$ if $n$ is odd and greater than $6$ and $2^n+3K_n(1)-16>0$ if $n$ is even and greater than 6. This demonstrates that $\tilde{\beta} (F)=4$ achieves the lowest value for a non-APN function when $3$ does not divide $n$, and $\tilde{\beta}(F)=8$ otherwise.
\end{remark}

We emphasize that it is generally a challenge to determine the Feistel boomerang uniformity for a given function, not to say its Feistel boomerang spectrum, see Table \ref{table-1}. By meticulously and carefully solving targeted equations over $\mathbb{F}_{2^n}$ and using the binary Kloosterman sum, we can determine the Feistel boomerang spectrum of $F$.

\subsection{Proof of part (A) of Theorem \ref{main-thm}}\label{proof1}

This subsection aims to prove part (A) of Theorem \ref{main-thm}. We shall use the results derived from the following statement.
 \begin{lemma}\label{lemma2-quar-root}\rm (\cite {LP.WK1972})
Let $F(x)=x^4+a_2x^2+a_1x+a_0$ with $a_0a_1\ne0$ and the companion cubic $G(y)=y^3+a_2y+a_1$ with the roots $r_1, r_2, r_3$. When the roots exist  in $\ftwon$, set $\omega_i=\frac{a_0r_i^2}{a_1^2}$. Let $h$ polynomial $h$ as $h=(1, 2, 3, \cdots)$ over some field to mean that it decomposes as a product of degree $1$, $2$, $3$, $\cdots$, over that field. The factorization of $F(x)$ over $\ftwon$  is characterized as follows:

\begin{itemize}
\item [\rm (1)] $F=(1, 1, 1, 1)$ corresponds to $G(1, 1, 1)$ and ${\rm Tr}_1^n(\omega_1)={\rm Tr}_1^n(\omega_2)={\rm Tr}_1^n(\omega_3)=0$;
\item [\rm (2)] $F=(2, 2)$ corresponds to  $G=(1, 1, 1)$ and ${\rm Tr}_1^n(\omega_1)=0$, ${\rm Tr}_1^n(\omega_2)={\rm Tr}_1^n(\omega_3)=1$;
\item [\rm (3)] $F=(2, 2)$  corresponds to $G=(3)$;
\item [\rm (4)] $F=(1, 1, 2)$ corresponds to  $G=(1, 2)$ and ${\rm Tr}_1^n(\omega_1)=0$;
\item [\rm (5)] $F=(4)$ corresponds to $G=(1, 2)$ and ${\rm Tr}_1^n(\omega_1)=1$.
\end{itemize}

\end{lemma}


According to Definition \ref{definition-FBCT}, it is sufficient to determine the number of solutions of the equation $F(x)+F(x+a)+F(x+b)+F(x+a+b)=0$ for $a,b \in \ftwon$ and $F(x)=x^{2^{n-2}-1}$,  i.e.,
  \begin{equation}\label{S1}
  x^{2^{n-2}-1}+(x+a)^{2^{n-2}-1}+(x+b)^{2^{n-2}-1}+(x+a+b)^{2^{n-2}-1}=0.
  \end{equation}

  \textbf{Case 1:} $ab(a+b)=0$. This is a trivial case, and one gets ${\rm FBCT}_F(a,b)=2^n.$

  \textbf{Case 2:} $ab(a+b)\neq0$. Let $y=\frac{x}{b}$ and $c=\frac{a}{b}$, where $c\neq0,1$. Then (\ref{S1}) is equivalent to
\begin{equation}\label{S2}
y^{2^{n-2}-1}+(y+c)^{2^{n-2}-1}+(y+1)^{2^{n-2}-1}+(y+c+1)^{2^{n-2}-1}=0.
\end{equation}

\textbf{Case 2.1:} If $y\in\{0, 1, c, c+1\}$, then (\ref{S2}) becomes
\begin{equation}\label{S3}
c^{2^{n-2}-1}+1+(c+1)^{2^{n-2}-1}=0.
\end{equation}

Since $c \neq 0,1$, we can multiply both sides of (\ref{S3}) by $c(c+1)$. This gives us $c^{2^{n-2}}=c^2$, which can be further simplified to $c^{2^{n-3}}=c$.
This means that $c\in {\mathbb F}_{2^{{\rm gcd}(n,n-3)}}$. If $3\nmid n$, that is, ${\rm gcd}(n,n-3)=1$, then we have $c\in\ftwo$, which contradicts with $c\neq0,1$. Therefore, when $3\nmid n$, $y=0,1,c,c+1$ are not solutions of (\ref{S2}). If $3\mid n$, then ${\rm gcd}(n,n-3)=3$. In this case, we find that $y=0,1,c,c+1$ are solutions of (\ref{S2}) when $c\in{\mathbb F}_{2^3}\setminus\ftwo$, and $y=0,1,c,c+1$ are not solutions of (\ref{S2}) when $c\in{\mathbb F}_{2^n}\setminus{\mathbb F}_{2^3}$.


\textbf{Case 2.2:} If $y\in\ftwon\setminus\{0,1,c,c+1\}$, multiplying $y(y+c)(y+1)(y+c+1)$ on both sides of (\ref{S2}) gives
 \begin{equation}\label{S4}
 (c^2+c)y^{2^{n-2}}+(c^{2^{n-2}}+c)y^2+(c^{2^{n-2}}+c^2)y=0.
 \end{equation}
Raising $4$-th power to (\ref{S4}) leads to
 \begin{equation}\label{S5}
 (c^4+c)y^8+(c^8+c)y^4+(c^8+c^4)y=0,
 \end{equation}
 which can be factorized as
 $$(c^2+c)y(y+1)(y+c)(y+c+1)((c^2+c+1)y^4+(c^4+c^2+1)y^2+(c^4+c)y+c^4+c^2)=0.$$
 Since $c^2+c\ne0,\ y\ne0,\ 1,\ c,\ c+1$, the above equation is equivalent to
 \begin{equation}\label{S6}
 (c^2+c+1)y^4+(c^4+c^2+1)y^2+(c^4+c)y+c^4+c^2=0.
 \end{equation}
 If $c^2+c+1=0$, (\ref{S6}) can be reduced to $c^4+c^2=(c^2+c)^2=1=0$, which is a contradiction. Then we have $c^2+c+1\neq0$, and (\ref{S6}) can be reduced to\\
  \begin{equation}\label{S7}
  y^4+(c^2+c+1)y^2+(c^2+c)y+\frac{c^4+c^2}{c^2+c+1}=0.
  \end{equation}
  Since $c^2+c\neq0$, we have $\frac{c^4+c^2}{c^2+c+1}\neq0$. By Lemma \ref{lemma2-quar-root}, the companion cubic polynomial of (\ref{S7}) is
  $$G(t)=t^3+(c^2+c+1)t+c^2+c,$$
 which can be factored as $(t+1)(t+c)(t+c+1)$ in $\ftwon$. If $G(t)=0$, we get $r_1=1,\ r_2=c,\ r_3=c+1$. Let $a_0=\frac{c^4+c^2}{c^2+c+1},\ a_1=c^2+c,\ a_2=c^2+c+1$,\ $\omega_1=\frac{a_0r_1^2}{a_1^2}=\frac{a_0}{a_1^2}=\frac{1}{c^2+c+1}$,\ $\omega_2=\frac{a_0r_2^2}{a_1^2}=\frac{a_0c^2}{a_1^2}=\frac{c^2}{c^2+c+1}$,\ $\omega_3=\frac{a_0r_3^2}{a_1^2}=\frac{a_0(c^2+1)}{a_1^2}=\frac{c^2+1}{c^2+c+1}$. Since $G(t)$ can be factored as $(1,1,1)$, by Lemma \ref{lemma2-quar-root},\ we can easily see that (\ref{S7}) has four solutions in $\ftwon$ if and only if
  ${\rm Tr}_1^n(\omega_1)={\rm Tr}_1^n(\omega_2)={\rm Tr}_1^n(\omega_3)=0$.

  Summarizing all cases, the  $\rm{FBCT}$ of $F(x)$ satisfies
  \begin{eqnarray}\label{eq-u=4}
 {\rm FBCT}_F(a,b) = \begin{cases}
		2^n, &  {\rm if}\,\, ab(a+b) = 0; \\
		4, &  {\rm if}\,\, ab(a+b)\neq 0,\,\, c^2+c+1\neq 0\\
           &  {\rm and}\,\, {\rm Tr}_1^n(\omega_1)={\rm Tr}_1^n(\omega_2)={\rm Tr}_1^n(\omega_3)=0; \\
		0, &  {\rm otherwise}, \end{cases}
  \end{eqnarray}
when {\rm 3} $\nmid n$, and for {\rm 3} $\mid n$, the {\rm FBCT} of $F(x)$ satisfies
 \begin{eqnarray}\label{eq-u=8}
  {\rm FBCT}_F(a,b) = \begin{cases}
		2^n, &  {\rm if}\,\, ab(a+b) = 0; \\
        8, &  {\rm if}\,\, ab(a+b)\neq 0,\,\, c^2+c+1\neq 0,\,\, \frac{a}{b}\in{\mathbb F}_{2^3}\\
           &  {\rm and}\,\, {\rm Tr}_1^n(\omega_1)={\rm Tr}_1^n(\omega_2)={\rm Tr}_1^n(\omega_3)=0;\\
		4, &  {\rm if}\,\, ab(a+b)\neq 0,\,\, c^2+c+1\neq 0,\,\, \frac{a}{b}\in{\mathbb F}_{2^n}\setminus{\mathbb F}_{2^3}\\
           &  {\rm and}\,\, {\rm Tr}_1^n(\omega_1)={\rm Tr}_1^n(\omega_2)={\rm Tr}_1^n(\omega_3)=0; \\
		0, &  {\rm otherwise}. \end{cases}
 \end{eqnarray}

  This completes the proof of part (A) of Theorem \ref{main-thm}.

\subsection{Proof of part (B) of Theorem \ref{main-thm}\label{proof2}}

In this subsection,  we will explore the proof of part (B) of Theorem \ref{main-thm}. We will discuss two lemmas about quadratic equations and an exponential sum connected to the Kloosterman sum over $\mathbb{F}_{2^n}$. These will be important for the subsequent discussions.
\subsubsection{Some auxiliaries  results}

\begin{lemma}\label{lemma1-qua-root}\rm (\cite{LR.NH1997})
 Let $a, b, c\in\ftwon, a\neq0$ and $F(x)=ax^2+bx+c$. Then
\begin{itemize}
\item [\rm (1)] $F(x)$ has exactly one root in $\ftwon$ if and only if $b=0$;
\item [\rm (2)] $F(x)$ has exactly two roots in $\ftwon$ if and only if $b\ne0$ and ${\rm Tr}_1^n(\frac{ac}{b^2})=0$,
\item [\rm (3)] $F(x)$ has no root in $\ftwon$ if and only if $b\ne0$ and ${\rm Tr}_1^n(\frac{ac}{b^2})=1$.
\end{itemize}
\end{lemma}

\begin{lemma}\label{lemma4-S}\rm
Let $n$ be a positive integer. Then
$$
S=\sum\limits_{x\in\ftwon}(-1)^{{\rm Tr}_1^n(\frac{x+1}{x^2+x+1})}=
\begin{cases}
   K_n(1)-2, &  {\rm if}\,\, n\,\, {\rm is\,\, odd} ;\\
   K_n(1), &  {\rm if}\,\, n\,\, {\rm is\,\, even}.\\
\end{cases}
$$
\end{lemma}

\begin{proof}
We calculate $S$ according to the parity of $n$ as below.

\textbf{Case 1:} If $n$ is odd, we have ${\rm Tr}_1^n(1)=1$ and then $x^2+x+1\ne0$ due to Lemma \ref{lemma1-qua-root}. Let $h=\frac{x+1}{x^2+x+1}$, then we have $hx^2+(h+1)x+h+1=0$. If $h=0$ or $h=1$, it has exactly one solution, namely, $x=1$, or $x=0$ respectively. For $h\ne0,1$,  again by Lemma \ref{lemma1-qua-root}, one has that it has two solutions if and only if ${\rm Tr}_1^n(\frac{h(h+1)}{h^2+1})=0$, which is equivalent to ${\rm Tr}_1^n(\frac{1}{h+1})=1$.
Then we have
$$S=(-1)^{{\rm Tr}_1^n(0)}+(-1)^{{\rm Tr}_1^n(1)}+2\sum_{h\in\ftwon\setminus\{0,1\},{\rm Tr}_1^n(\frac{1}{h+1})=1}(-1)^{{\rm Tr}_1^n(h)}.$$
Note that
$$2\sum_{h\in\ftwon\setminus\{0,1\},{\rm Tr}_1^n(\frac{1}{h+1})=1}(-1)^{{\rm Tr}_1^n(h)}
=2\sum_{h\in\ftwon\setminus\{0,1\},{\rm Tr}_1^n(\frac{1}{h})=1}(-1)^{{\rm Tr}_1^n(h+1)}$$
which leads to
$$S=-2\sum_{h\in\ftwon\setminus\{0,1\},{\rm Tr}_1^n(\frac{1}{h})=1}(-1)^{{\rm Tr}_1^n(h)}=-2(\sum_{h\in\ftwon,{\rm Tr}_1^n(\frac{1}{h})=1}(-1)^{{\rm Tr}_1^n(h)}+1).$$
Observe that
$$
\begin{aligned}
\sum_{h\in\ftwon,{\rm Tr}_1^n(\frac{1}{h})=1}(-1)^{{\rm Tr}_1^n(h)}&=\sum_{h\in\ftwon,{\rm Tr}_1^n(\frac{1}{h})=1,{\rm Tr}_1^n(h)=0}1-\sum_{h\in\ftwon,{\rm Tr}_1^n(\frac{1}{h})=1,{\rm Tr}_1^n(h)=1}1\\
&=\sum_{h\in\ftwon,{\rm Tr}_1^n(\frac{1}{h})=1}1-2\sum_{h\in\ftwon,{\rm Tr}_1^n(\frac{1}{h})=1,{\rm Tr}_1^n(h)=1}1\\
&=2^{n-1}-2|\Phi|,
\end{aligned}
$$
where $\Phi=\{h\in\ftwon:{\rm Tr}_1^n(\frac{1}{h})=1,{\rm Tr}_1^n(h)=1\}$ and  satisfies
$$
\begin{aligned}
4|\Phi|&=\sum_{h\in\ftwon}\sum_{u_1\in\ftwo}(-1)^{u_1(1+{\rm Tr}_1^n(\frac{1}{h}))}\sum_{u_2\in\ftwo}(-1)^{u_2(1+{\rm Tr}_1^n(h))}\\
&=\sum_{h\in\ftwon}\Big(1+(-1)^{1+{\rm Tr}_1^n(\frac{1}{h})}\Big)\Big(1+(-1)^{1+{\rm Tr}_1^n(h)}\Big)\\
&=\sum_{h\in\ftwon}\Big(1-(-1)^{{\rm Tr}_1^n(\frac{1}{h})}-(-1)^{{\rm Tr}_1^n(h)}+(-1)^{{\rm Tr}_1^n(\frac{1}{h}+h)}\Big)\\
&=2^n+K_n(1).
\end{aligned}
$$
It is clear that $2|\Phi|=2^{n-1}+\frac{1}{2}K_n(1)$. Then we have
$$\sum_{h\in\ftwon,{\rm Tr}_1^n(\frac{1}{h})=1}(-1)^{{\rm Tr}_1^n(h)}=-\frac{1}{2}K_n(1)$$
and the desired result follows.

\textbf{Case 2:}
If $n$ is even, then ${\rm Tr}_1^n(1)=0$ and $x^2+x+1=0$ has two solutions by Lemma \ref{lemma1-qua-root}. Note that $1/0:=0$. Then, similarly to the case of $n$ is odd, let $g=\frac{x+1}{x^2+x+1}$, one obtains
$$
\begin{aligned}
S&=3\cdot(-1)^{{\rm Tr}_1^n(0)}+(-1)^{{\rm Tr}_1^n(1)}+2\sum\limits_{g\in\ftwon\setminus\{0,1\},{\rm Tr}_1^n(\frac{1}{g+1})=0}(-1)^{{\rm Tr}_1^n(g)}\\
&=4+2\sum\limits_{g\in\ftwon\setminus\{0,1\},{\rm Tr}_1^n(\frac{1}{g})=0}(-1)^{{\rm Tr}_1^n(g+1)}\\
&=4+2\sum\limits_{g\in\ftwon\setminus\{0,1\},{\rm Tr}_1^n(\frac{1}{g})=0}(-1)^{{\rm Tr}_1^n(g)}.\\
\end{aligned}
$$
It is not difficult to see that
$$
\begin{aligned}
\sum\limits_{g\in\ftwon\setminus\{0,1\},{\rm Tr}_1^n(\frac{1}{g})=0}(-1)^{{\rm Tr}_1^n(g)}&=\sum\limits_{g\in\ftwon\setminus\{0,1\}}(-1)^{{\rm Tr}_1^n(g)}-\sum\limits_{g\in\ftwon\setminus\{0,1\},{\rm Tr}_1^n(\frac{1}{g})=1}(-1)^{{\rm Tr}_1^n(g)}.\\
\end{aligned}
$$
By the balance property of the trace function and the discussion in Case 1, we have
\[\sum\limits_{g\in\ftwon\setminus\{0,1\}}(-1)^{{\rm Tr}_1^n(g)}=-2,\;\;\sum\limits_{g\in\ftwon\setminus\{0,1\},{\rm Tr}_1^n(\frac{1}{g})=1}(-1)^{{\rm Tr}_1^n(g)}=-\frac{1}{2}K_n(1)\]
which indicates that $S=4+2(-2-(-\frac{1}{2}K_n(1)))=K_n(1)$.\\
This completes this proof of Lemma \ref{lemma4-S}.
\end{proof}

\subsubsection{Proof of part (B) of Theorem \ref{main-thm}}

We are ready to present the proof of part (B) of Theorem \ref{main-thm}.


From the part (A) of Theorem \ref{main-thm}, it is sufficient to determine the values of
$$\Theta_i=|\{(a,b)\in \ftwon^2:{\rm FBCT}_F(a,b)=i\}|$$
for $i=2^n,0,4,8$. Clearly, $\Theta_{2^n}=3\times2^n-2$ since ${\rm FBCT}_F(a,b)=2^n$ if and only if $ab(a+b)=0$.

According to \eqref{eq-u=4} and \eqref{eq-u=8}, we then proceed with the proof as follows:

\textbf{Case 1: } $3\nmid n$.

Recall that $c=\frac{a}{b}$, $w_1=\frac{1}{c^2+c+1}$, $w_2=\frac{c^2}{c^2+c+1}$ and $w_3=\frac{c^2+1}{c^2+c+1}$. Then by \eqref{eq-u=4} and the fact $w_3=w_1+w_2$, one can conclude that $\Theta_4=(2^n-1)|D|$, where $D$ is given by
 $$D=\{c\in\ftwon\setminus\{0,1\}:  c^2+c+1\ne 0, {\rm Tr}_1^n(\frac{1}{c^2+c+1})=0, {\rm Tr}_1^n(\frac{c^2}{c^2+c+1})=0\}.$$

To compute the cardinality of $D$, define
\begin{eqnarray*}
  D_1&=&\{c\in\ftwon: {\rm Tr}_1^n(\frac{1}{c^2+c+1})=0, {\rm Tr}_1^n(\frac{c^2}{c^2+c+1})=0\} \\
  D_2&=&\{c\in\{0,1\}: {\rm Tr}_1^n(\frac{1}{c^2+c+1})=0, {\rm Tr}_1^n(\frac{c^2}{c^2+c+1})=0\}\\
  D_3&=&\{c^2+c+1=0: {\rm Tr}_1^n(\frac{1}{c^2+c+1})=0, {\rm Tr}_1^n(\frac{c^2}{c^2+c+1})=0\}
\end{eqnarray*}
and correspondingly, one gets
$$|D|=|D_1|-|D_2|-|D_3|.$$
By Lemma \ref{lemma1-qua-root} and the definition of the trace function, one can obtain $|D_3|=|D_2|=0$ if $n$ is odd and otherwise $|D_3|=|D_2|=2$. Thus, we have $|D|=|D_1|$ if $n$ is odd and $|D|=|D_1|-4$ if $n$ is even. Observe that
\begin{eqnarray*}
4|D_1|&=&\sum_{c\in\ftwon}\sum_{u_1\in\ftwo}(-1)^{u_1{\rm Tr}_1^n(\frac{1}{c^2+c+1})}\sum_{u_2\in\ftwo}(-1)^{u_2{\rm Tr}_1^n(\frac{c^2}{c^2+c+1})} \\
&=&2^n+\sum_{c\in\ftwon}(-1)^{{\rm Tr}_1^n(\frac{1}{c^2+c+1})}+\sum_{c\in\ftwon}(-1)^{{\rm Tr}_1^n(\frac{c^2}{c^2+c+1})}+\sum_{c\in\ftwon}(-1)^{{\rm Tr}_1^n(\frac{c^2+1}{c^2+c+1})} \\
&=&2^n+\sum_{c\in\ftwon}(-1)^{{\rm Tr}_1^n(\frac{1}{c^2+c+1})}+(-1)^{{\rm Tr}_1^n(1)}\Big(\sum_{c\in\ftwon}(-1)^{{\rm Tr}_1^n(\frac{c+1}{c^2+c+1})}+\sum_{c\in\ftwon}(-1)^{{\rm Tr}_1^n(\frac{c}{c^2+c+1})}\Big).
\end{eqnarray*}
Further, we have
\begin{eqnarray*}
 \sum_{c\in\ftwon}(-1)^{{\rm Tr}_1^n(\frac{1}{c^2+c+1})}&=&\sum_{\frac{1}{c}\in\ftwon^\ast}(-1)^{{\rm Tr}_1^n(\frac{1}{\frac{1}{c^2}+\frac{1}{c}+1})}+(-1)^{{\rm Tr}_1^n(1)}\\
 &=&(-1)^{{\rm Tr}_1^n(1)}\sum_{c\in\ftwon}(-1)^{{\rm Tr}_1^n(\frac{c+1}{c^2+c+1})}-1+(-1)^{{\rm Tr}_1^n(1)},\\
 \sum\limits_{c\in\ftwon}(-1)^{{\rm Tr}_1^n(\frac{c}{c^2+c+1})}  &=& \sum\limits_{c\in\ftwon}(-1)^{{\rm Tr}_1^n(\frac{c+1}{(c+1)^2+(c+1)+1})}= \sum\limits_{c\in\ftwon}(-1)^{{\rm Tr}_1^n(\frac{c+1}{c^2+c+1})}.
\end{eqnarray*}
This together with Lemma \ref{lemma4-S} implies that $|D|=|D_1|=(2^n-3K_n(1)+4)/4$ if $n$ is odd and otherwise $|D|=|D_1|-4=(2^n+3K_n(1)-16)/4$.
Then by $\Theta_4=(2^n-1)|D|$, one gets
\begin{eqnarray}\label{eq-theta4}
\Theta_4=
\begin{cases}
   \frac{(2^n-1)(2^n-3K_n(1)+4)}{4}, &  {\rm if}\,\, n\,\, {\rm is\,\, odd} ;\\
   \frac{(2^n-1)(2^n+3K_n(1)-16)}{4}, &  {\rm if}\,\, n\,\, {\rm is\,\, even},\\
\end{cases}
\end{eqnarray}
and consequently, by \eqref{eq-u=4}, one obtains
$$
\Theta_0=
\begin{cases}
   \frac{(2^n-1)(3\cdot2^n+3K_n(1)-12)}{4}, &  {\rm if}\,\, n\,\, {\rm is\,\, odd} ;\\
   \frac{(2^n-1)(3\cdot2^n-3K_n(1)+8)}{4}, &  {\rm if}\,\, n\,\, {\rm is\,\, even}.\\
\end{cases}
$$

\textbf{Case 2: } $3\mid n$.

For this case, by \eqref{eq-u=4}, \eqref{eq-u=8} and \eqref{eq-theta4}, one can conclude that
\begin{eqnarray*}
\Theta_4+\Theta_8=
\begin{cases}
   \frac{(2^n-1)(2^n-3K_n(1)+4)}{4}, &  {\rm if}\,\, n\,\, {\rm is\,\, odd} ;\\
   \frac{(2^n-1)(2^n+3K_n(1)-16)}{4}, &  {\rm if}\,\, n\,\, {\rm is\,\, even}.\\
\end{cases}
\end{eqnarray*}

Building on our previous discussions and referring again to \eqref{eq-u=8}, we can derive that $\Theta_8=(2^n-1)|D|$. This simplifies to $\frac{(2^n-1)(2^3-3K_3(1)+4)}{4}=6(2^n-1)$, given that $K_3(1)=-4$. This derivation allows us to determine the values of $\Theta_4$ and $\Theta_0$.\\
This completes the proof of  (B) of Theorem \ref{main-thm}.

\section{Conclusion}\label{conc}

The Feistel Boomerang Connectivity Table ($\rm{FBCT}$) is a crucial tool for analyzing the security of block ciphers against powerful attacks, such as boomerang attacks. In our research, we focused on the  $\rm{FBCT}$ for the power function $F(x)=x^{2^{n-2}-1}$ over $\mathbb{F}_{2^n}$, where $n>6$ is an integer. We performed detailed manipulations to solve certain equations over $\mathbb{F}_{2^n}$. We used the value of the binary Kloosterman sum at point $1$ to calculate the values of entries in its  $\rm{FBCT}$. Additionally, we determined its Feistel boomerang spectrum.  We emphasise that the function $F$ achieves the lowest Feistel boomerang uniformity among non-APN functions provided that $n$ is not divisible by $3$.

\section*{Acknowledgements}

This work was supported by the National Natural Science Foundation of China (Nos. 62072162 and 12001176), the innovation group project of the Natural Science Foundation of Hubei Province of China (No. 2023AFA021), the Natural Science Foundation of Hubei Province of China (No. 2021CFA079) and the Knowledge Innovation Program of Wuhan-Basic Research (No. 2022010801010319).

\end{document}